%% file: main.tex
\documentclass[twoside]{article}

\usepackage[preprint]{aistats2027}
\usepackage[round]{natbib}

\input{preamble.tex}

\begin{document}

\twocolumn[

\aistatstitle{Kernel-based External Validity of Experimental Studies}

\aistatsauthor{Federico Matteucci \And Gerhard Neumann}

\aistatsaddress{Karlsruhe Institute of Technology \And Karlsruhe Institute of Technology}]

\begin{abstract}
The external validity of a study is the extent to which its findings generalize to unseen experimental conditions, such as datasets.
We formalize the results of a study together with its research question, and define external validity as the probability that two independent and same-sized replications of the study produce results whose maximum mean discrepancy (MMD) is below a given tolerance.
This approach, centered on the choice of a space of results and of an appropriate reproducing kernel for the MMD, is flexible enough to be deployed on multiple different types of results and research questions.
Additionally, we establish an asymptotic power-law relation between the size of a study and the quantiles of the MMD, which we use to estimate the desired size of an externally valid study from a small-scale preliminary study.
We deploy our framework to rankings with ties and discuss kernels suited to different research questions.
Experiments on existing case studies and on a simulated newly designed study illustrate the applicability of the framework.
We release our code as part of the \texttt{genexpy} Python module.\textsuperscript{\ref{fn:code}}
\end{abstract}

\section{INTRODUCTION}
\label{sec:introduction}

Assessing the quality of an experimental study is a multi-faceted problem and depends on many independent aspects. 
Among these, its scope should be properly defined~\citep{boulesteix2015hypothesis, bouthillier2021variance}, its documentation should allow for independent reproducibility studies~\citep{baker_is_2016, peng2011reproducible, raff2021reproducibility, raff2023citations}, its results should be sensibly analyzed to assess their significance~\citep{benavoli2017time, corani2017statistical, demsar2006comparison} and should be valid for a population extending beyond the original sample~\citep{liao2021are}.
External validity, which is sometimes referred to as reproducibility, replicability, or generalizability, is this last component.

A study is deemed externally valid if repeating it under comparable experimental conditions, and in the absence of systematic changes to its design, yields similar results~\citep{nas2019reproducibility, findley2021external, pineau2021reproducibility}.
External validity is assessed through the sign- and effect-generalization of a treatment on a numeric response in causal inference~\citep{campbell_factors_1957, egami2023external}, psychology~\citep{opensciencecollaboration2015reproducibility, klein2018manylabs2, patil2016what, verhagen2014bayesian}, economics~\citep{camerer2016replicability}, the social sciences~\citep{delios2022examining, camerer2018evaluating}, medicine~\citep{dersimonian1986meta}, human-AI interaction~\citep{vaccaro2024combinations}, and chain-of-thought reasoning~\citep{sprague2025metanalysis}.
Generalizability theory~\citep{cronbach1963Gtheory, brennan1992generalizability} quantifies the contribution of covariates to the variance of a numeric response.

It is well documented that significant findings often fail to hold beyond the study that produced them.
In machine learning (ML) specifically, the results of experimental studies are well known to vary substantially with design choices~\citep{boulesteix2017evidence, gundersen2023sources, vanmechelen2023benchmark, lu2023coreset, qin2023rdsuite, mcelfresh2022recommender, wainberg2016are}, be it an insufficient number of datasets~\citep{dehghani2021benchmark, alvarez2022revealing, boulesteix2015hypothesis} or differences in hyperparameter tuning~\citep{bouthillier2021variance}, initialization seed~\citep{gundersen2023reporting}, hardware~\citep{zhuang2022neural}, minor perturbations in the data~\citep{alzahrani2024benchmarks}, or the inclusion of additional competitors~\citep{zhang2024inherent}.
We refer to \citet{liao2021are} for a broader discussion.
This motivates the search for alternative methods to assess the external validity of findings that are not restricted to numeric responses and apply to other types of data, such as rankings.

A natural tool for this purpose is the maximum mean discrepancy (MMD), which compares distributions by mapping them into an appropriate reproducing kernel Hilbert space (RKHS).
Successful applications include testing \citep{gretton2012kernel,fernandez2019maximum,liu2020learning,wynne2022kernel,chevyrev2022signature,schrab2023mmd,zhou2025dual, chatalic2026nystroemmmd}, generative modeling \citep{dziugaite2015training,li2015generative,sutherland2017generative,li2017mmd,zhou2025inductive}, quantization \citep{arbel2019maximum,gladin2024interaction,belhadji2026weighted}, subspace selection \citep{cribeiroramallo2025adversarial}, robust regression \citep{alquier2024universal}, learning world models \citep{balestriero2025lejepa,zimmermann2025kerjepa}, model selection \citep{bruck2025distribution}, and online change detection \citep{li2015mstatistic,li2019scan,kalinke2025maximum,kalinke2025optimal}.
At the core of the MMD's flexibility are reproducing kernels, which have been defined on a variety of domains, including strings and graphs \citep{borgwardt2006integrating,obray2022evaluation}, Riemannian manifolds \citep{cheng2024kernel}, Hilbert spaces \citep{wynne2022kernel}, time series and stochastic processes \citep{chwialkowski2014wild,laumann2021kernel,chevyrev2022signature}, persistence diagrams \citep{kwitt2015statistical}, and (partial) rankings \citep{mania2018kernelrankings, lomeli2019antithetic}.
Rankings are particularly prominent in the analysis of experimental results, especially in non-parametric tests and procedures~\citep{demsar2006comparison, benavoli2017time}.
This work aims at applying the MMD to the issue of lack of external validity. 

Our core contribution is a novel take on the problem of assessing the external validity of an experimental study. 
To this end, we propose a kernel-based framework for the formalization and analysis of experimental results, with only mild structural assumptions, and translate an existing natural-language definition of external validity into our statistic, which we call kernel external validity (KEV). 
We investigate the properties of our statistic and discuss how to set its parameters (the kernel $k$ and a numeric threshold $\eps$) in an interpretable and informed way.
We discuss how the KEV can help assess the appropriate size of an experimental study, and apply it to a simulated study to show how it can be useful in practice. 
We show applications of the KEV to rankings and two case studies.  
Finally, we release our code as part of the \texttt{genexpy}\footnote{\label{fn:code} \mbox{\url{https://github.com/DrCohomology/genexpy}}.} Python module, which contains routines to handle rankings, implementations of kernels for rankings, and the MMD.

\textbf{Outline}\; \cref{sec:background} introduces our notation, formalizes experimental studies, and recalls the MMD.
\cref{sec:externalvalidity} discusses our KEV statistic and how to choose its parameters. 
\cref{sec:experiments} contains our experiments and case studies.
\cref{sec:conclusions} details limitations.

\section{BACKGROUND}
\label{sec:background}

This section introduces our notation, formalizes experimental studies, and recalls the MMD.
We collect other section-specific notation where necessary.

\subsection{Notation}

Our notation is as follows.

\textbf{General}\;
We denote with $\N$ the set of positive integers and with $\R^+$ the real interval $(0, +\infty)$.
We denote with $[n]$ the set $\{1, \dots, n\} \subset \N$.
For sets $A, B$, $A \sqcup B$ indicates their disjoint union.
For a matrix $M$, let $M^\transpose$ be its transpose.
Given a set $\m X$, $\b v \in \m X^n$, and $\b v' \in \m X^m$, $\b v|| \b v' \in \m X^{n+m}$ is their concatenation.
Given sets $\m X$ and $\m Y$, $\m Y^{\m X}$ is the set of functions $\m X \to \m Y$.
For a function $f \in \m Y^{\m X}$, $f(x)\equiv y$ is equivalent to $f(x) = y$ for all $x\in \m X$.
A ranking (with ties) $\ra$ of $\nalt$ alternatives is a total and transitive binary endorelation on $[\nalt]$.
For $\alt \in [\nalt]$, $\ra(\alt)$ is the rank of $\alt$ in $\ra$.
The set of rankings of $\nalt$ alternatives is denoted with $\Ra$. 

\textbf{Topology and Probability}\;
Let $(\m X, \tau_{\m X})$ be a separable metrizable topological space. 
We write $\m M_1^+(\m X)$ for the set of probability measures on $\m X$, understood w.r.t.\ the measurable space $(\m X, \m B_{\m X})$, where $\m B_{\m X}$ denotes the Borel $\sigma$-algebra induced by $(\m X, \tau_{\m X})$. 
For $x\in\m X$, the Dirac measure is written as $\delta_x$, that is, $\delta_x(A) = 1$ if $x\in A$ and $0$ otherwise ($A\in \m B_{\m X}$).
We use $P \otimes P$ and $P^2$ interchangeably, depending on the number of random variables; for instance, we write $X, X' \sim P \otimes P$ and $W \sim P^2$.
For $f \in \m Y^{\m X}$ measurable, $\E_{X\sim P} f(X) = \int_{\m X} f(x) \d P(x)$; if there is no ambiguity, we shorten it to $\E_P f(X)$.
We indicate with $F_X$ the cumulative distribution function (cdf) of a real-valued random variable $X$.
For $n \in \N$ and $\b x \in \m X^n$, we denote with $\hat P_{\b x} \coloneqq n^{-1} \sum_{i = 1}^n \delta_{x_i}$; moreover, we define the function $\hat P: \b x \mapsto \hat P_{\b x}$.
Given a measurable function $X: (\Omega, \m B_\Omega, \P) \to (\m X, \m B_\m X)$, the pushforward measure of $\P$ by $X$ is $X_*\P: A\mapsto \P(X^{-1}(A))$, for $A \in \m B_\m X$.

\textbf{(Reproducing Kernel) Hilbert Spaces}\;
Let $k: \m X \times \m X \to \R$ be a measurable kernel with separable RKHS $(\m H_k, \lrang{\cdot, \cdot}_{\m H_k})$.
In other words, $\m H_{k}$ is a space of functions $\m X \to \R$ for which $k$ has the reproducing property: $\forall h \in \m H_k$, $h(x) = \lrang{h, k(\cdot, x)}_{\m H_k}$.
The canonical feature map associated to $k$ is, for $x\in \m X$, $k(\cdot, x) \in \m H_k$.
The mean embedding of a distribution $P \in \m M_1^+(\m X)$ relative to $k$ is $\mu_P \coloneqq \E_P k(\cdot, X) \in \m H_k$, provided that $\E_{P}\lrn{k(\cdot,X)}_{\m H_k} < \infty$; the latter is guaranteed if $\sup_{x\in\m X}\sqrt{k(x,x)} < +\infty$, in which case $k$ is called bounded.
A kernel $k$ has constant diagonal if $k(x, x) \equiv c \in \R$ and is normalized if $c=1$.

\subsection{Experimental Studies}
\label{sec:background:studies}

\begin{figure}[thb]
    \centering
    \includegraphics[width=\linewidth]{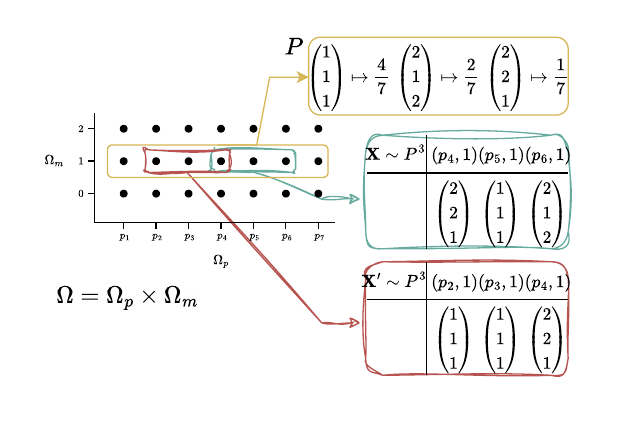}
    \caption{\cref{ex:checkmate_study} as an experimental study. We are interested in the external validity of the results $P$ for a fixed configuration $m=1$.}
    \label{fig:checkmate_study}
\end{figure} 

An (experimental) study is a collection of experiments evaluating the same alternatives (such as the LLMs in \cref{ex:checkmate_study}) under different experimental conditions.
We assume a factorial design, i.e., that every experiment is uniquely identified by an experimental condition and that the levels of the random factors can be chosen independently. 

\textbf{Mathematical Formalization}\;
For a countable set of indices $I$, a factor is $i \in I$ and $\Omega_i$ is the set of levels that factor $i$ can take.
For any $i$, we endow $\Omega_i$ with a topology $\tau_{\Omega_i}$, the Borel $\sigma$-algebra $\m B_{\Omega_i}$, and a probability measure $\P_i \in \m M_1^+(\Omega_i)$.
We define the probability space of experimental conditions as $\Omega \coloneqq \prod_{i \in I} \Omega_i$, with the product $\sigma$-algebra $\m B$ and probability $\P = \bigotimes_{i\in I}\P_i$.
We partition the factors into fixed and random ones: the generalization claim concerns different levels of the random factors.
Formally, letting $I_\fixed$ and $I_\random$ be the sets of fixed and random factors respectively, we have $I = I_\fixed \sqcup I_\random$ and $\Omega = \prod_{I_\fixed} \Omega_i \times \prod_{I_\random} \Omega_i \eqqcolon \Omega_\fixed \times \Omega_\random$. 
We define a configuration as $\omega_\fixed \in \Omega_\fixed$.
The (experimental) result is a random variable $X: \Omega \to \m X$, whose distribution we indicate with $P \coloneqq X_*\P \in \m M_1^+(\m X)$.

The following remark elaborates on some considerations of practical interest.
\begin{remark}\label{remark:practicalconsiderations}~
    \begin{enumerate}[label=(\alph*)]
        \item Running experiments means sampling an experimental condition $\omega$ from $\P$ and then computing its result $X(\omega)$.
        \item For an experimenter, choosing appropriate $\P_i$-s is part of the experimental setup, as certain levels might be more likely to occur.
        If $\Omega_i$ is finite or a bounded real interval and if no prior knowledge is available, $\P_i$ can be taken to be a uniform distribution. 
        If $\Omega_i$ is unbounded, the experimenter might want to restrict it to a bounded set; the claim of generalization concerns this bounded set. 
        \item Throughout this work, we evaluate the external validity of a study for a fixed configuration and randomly sampled levels of the random factors. 
        \item \cref{app:studies} gives a finer-grained classification of factors following~\citet[Chapter~1]{montgomery2017design}.
    \end{enumerate}
\end{remark}

\begin{example}(The ``checkmate-in-one'' task, \cref{fig:checkmate_study})
\label{ex:checkmate_study}~
    We compare three large language models (LLMs), the alternatives, on the ``checkmate-in-one'' task~\citep{srivastava2023beyond}, where an LLM succeeds if it outputs the unique checkmating move from a position of pieces on a chessboard.
    The factors involved are the number of shots $m$ (fixed factor) and the pieces' initial position $p$ (random factor), which is sampled uniformly at random from a finite pool.
    The research question is the following: \textit{``Is \textnormal{$\texttt{LLM}_1$} consistently ranked the same under different positions?''}
    The results are rankings, $\m X=\m R_3$, where an LLM achieves rank $1$ if it successfully solves the problem and $2$ otherwise.
\end{example}

\subsection{Maximum Mean Discrepancy}
\label{sec:background:mmd}

For a bounded and measurable kernel $k$, for $P,Q\in\m M_1^+(\m X)$, we define $\MMD(P,Q) \coloneq \lrn{\mu_P-\mu_Q}_{\m H_k}$.
Given two random samples $\b X = (X_1, \dots, X_n)^\transpose \sim P^n$ and $\b X' = (X'_1, \dots, X'_m)^\transpose \sim P^m$, the MMD between their respective empirical distributions takes the form
\begin{align}
\label{eq:mmd_empirical}
    &\MMD(\hat P_\b X, \hat P_{\b X'})^2 
        = \frac{1}{n^2} \sum\limits_{i, j=1}^n k(X_i, X_j) \\
        &\quad+ \frac{1}{m^2} \sum\limits_{i, j=1}^m k(X'_i, X'_j)
        - \frac{2}{nm} \sum\limits_{i, j=1}^{n,m} k(X_i, X'_j),
\end{align}
which is the same form as the V-statistic estimator for $\MMD(P, P) = 0$.
We will leverage this structural equivalence in the proof of \autoref{thm:mmd_limit}.

We define the function $\MMD_n$ as the composition 
\begin{equation}
\label{eq:mmdn_def}
\begin{tikzcd}[arrow style=tikz, >={Straight Barb[scale=0.8]}, column sep=2.5em]
\Omega^{n+n}
  \arrow[r, "X^n \times X^n"]
  \arrow[rrr, bend right=15, "\MMD_n"']
& \m X^{n+n} \arrow[r, "\hat P \times \hat P"]
& \m M_1^+(\m X)^2 \arrow[r, "\MMD"]
& \R
\end{tikzcd}.
\end{equation}
Under \cref{ass:well-defined-KEV} (i) and (ii), $\MMD_n$ is a real-valued random variable.

\section{KERNEL EXTERNAL VALIDITY}
\label{sec:externalvalidity}

In this section, we detail how to translate a standard definition of external validity into a kernelized mathematical framework and analyze its properties. 
We begin by listing our assumptions.
\begin{assumption}\label{ass:well-defined-KEV}
    (i) $(\m X,\tau_{\m X})$ is a separable metrizable topological space.
    (ii) $k: \m X \times \m X \to \R$ is a bounded and continuous kernel.
    (iii) The experiments follow a factorial design; $(\Omega, \tau_\Omega)$ is Polish, and the levels of the random factors are sampled independently of the observed results, as described in \cref{sec:background:studies}.
    (iv) We compare finitely many alternatives under countably many experimental factors.
    (v) Every alternative is evaluated under every experimental condition.
\end{assumption}

\begin{remark}\label{remark:well-defined-KEV}~
    \begin{enumerate}[label=(\alph*)]
        \item (i) and (ii) allow the definition of the $\MMD$ and make $\MMD_n$ measurable.
        \item (ii) implies that $H_k$ is separable.
        \item (iii) and (iv) allow i.i.d. samples from $\P$ and $P$ as well as the incremental sampling of \cref{alg:studyexecution}.
        \item (v) simplifies the analysis.\footnote{In practice and in \cref{sec:experiments}, we impute the missing evaluations.}
    \end{enumerate}
\end{remark}

After discussing our initial assumptions, we are ready to introduce our KEV statistic.
Recall that a study is externally valid if repeating it under comparable experimental conditions yields similar results~\citep{nas2019reproducibility, pineau2021reproducibility}. 
Our statistic translates this definition into a kernel framework: comparable experimental conditions are independent draws of the levels of the random factors, and similar results are results whose MMD is at most $\eps$.
\begin{definition}[KEV statistic]
\label{def:KEVstatistic}
    Under \cref{ass:well-defined-KEV}, with $\MMD_n$ as in \eqref{eq:mmdn_def}, and for $\eps \in \R^+$, the KEV of a study with results $P \in \m M_1^+(\m X)$ is 
    \begin{align}
    \label{eq:KEV_def}
        \ev_n\lrp{P, \eps}
            &\coloneqq F_{\MMD_n}(\eps)\notag\\
            &\,= P^n \otimes P^n \lrp{\MMD\lrp{\hat P_{\b X}, \hat P_{\b X'}} \leq \eps}.
    \end{align}
    We say that a study with results $P$ is $(\alpha, \eps, n)$-(externally) valid if $\ev_n\lrp{P, \eps} \geq \alpha$.
\end{definition}
As discussed in \cref{remark:practicalconsiderations} (c), we are usually interested in the external validity of a study for a fixed configuration.

We remark that the KEV is not pointwise monotonic in $n$ due to steps in the cdf of the MMD for discrete spaces, a problem especially noticeable when $n$ is small.
However, the KEV is monotonic in the sense described by the following lemma, proved in \cref{app:proofs:KEVmonotonicity}.
\begin{lemma}[Monotonicity of the KEV]
\label{lemma:KEVmonotonicity}
    Under \cref{ass:well-defined-KEV}, we have
    \begin{align}
        \int_\R \lrp{\ev_{n+1}(P, y) - \ev_{n}(P, y)} \,\d y \geq 0
    \end{align}
\end{lemma}
We now proceed to discuss the parameters of \cref{def:KEVstatistic}, starting with the kernel $k$.

\textbf{Kernels Formalize Research Questions}\; 
The external validity of a study ultimately depends on its research question, that is, what features of the results are relevant. 
Recalling \cref{ex:checkmate_study}, two research questions might be \emph{``Is $\texttt{LLM}_1$ consistently in the same position of the ranking?''} and \emph{``Are the best alternatives consistently the same ones?''}, which consider as the only relevant features the rank of $\texttt{LLM}_1$ and the top tier, respectively.
Kernels offer a clean formalization of this concept: the feature map $k(\cdot, x)$ keeps only the relevant features, as shown in the following example.
\begin{example}
    Assume $\m X = \lrc{x, x'}$ such that $k(\cdot, x) = k(\cdot, x')$.
    For any $P \in \m M_1^+(\m X)$, $\mu_P = P(\{x\})\,k(\cdot,x) + P(\{x'\})\,k(\cdot,x') = k(\cdot,x)$. 
    Hence, for any $Q \in \m M_1^+(\m X)$, we have $\mu_P - \mu_Q = 0_{\m H_k}$ and thus $\MMD(P, Q) \equiv 0$.
    In particular, this holds for $\hat P_{\b X}, \hat P_{\b X'}$ in \eqref{eq:KEV_def}: $\ev_n(P, \eps) = 1$ for any choice of $P$ and $\eps$.
\end{example}
\begin{remark}
    We do not require kernels to be characteristic; indeed, controlling what features are discarded by the kernel is a core component of our framework.
\end{remark}

The next section introduces kernels for rankings and their research questions.

\subsection{Kernels for Rankings}
\label{sec:externalvalidity:kernels:kernelsforrankings}

The \emph{Borda kernel} $\kernel[b]$ is suitable for research questions of the form \emph{``Is alternative $a^*$ consistently ranked the same?''}.
For $a^* \in [\nalt]$ and $\nu \in \R^+$,
\begin{equation}
\label{eq:borda_kernel}
    \kernel[b][\alt^*, \nu]\lrp{\ra_1, \ra_2} =  e^{-\nu \lra{b_1 - b_2}} \quad \in \lrb{e^{-\nu (\nalt-1)}, 1},
\end{equation}
where $b_l = \lra{\lrc{\alt\in[\nalt]: \ra_l(\alt)
\geq \ra_l(\alt^*)}}$ is the Borda count~\citep{borda1781memoire}. 

The \emph{Jaccard kernel} $\kernel[j]$ is suitable for research questions of the form \emph{``Are the top-$t$ alternatives consistently the same ones?''}.
We use the Jaccard similarity $J_t$, which is known to be a kernel; see e.g. \citet{bouchard2013jaccard}.
For a tier $t \in \N$, 
\begin{equation}
\label{eq:jaccard_kernel}
\kernel[j][t]\lrp{\ra_1, \ra_2} \coloneqq J_t(\ra_1, \ra_2) = \frac {\lra{\ra_{1, [t]} \cap \ra_{2, [t]}}} {\lra{\ra_{1, [t]} \cup \ra_{2, [t]}}} \quad\in [0, 1],
\end{equation}
where (for $l = 1, 2$) $\ra_{l, [t]} \coloneqq \lrc{\alt\in[\nalt]: \ra_l(\alt) \leq t}$ is the set of alternatives with rank at most $t$.

Finally, the \textit{Mallows kernel} $\kernel[m]$ is suitable for research questions of the form \emph{``Are the alternatives ranked consistently?''}
We adapt the existing kernel for untied and partial rankings~\citep{jiao2015kendall, mania2018kernelrankings, mallows1957model, lomeli2019antithetic} to rankings with ties.
For $\nu \in \R^+$,
\begin{equation}
\label{eq:mallows_kernel}
\kernel[m][\nu](\ra_1, \ra_2) = e^{-\nu n_d} \quad \in \lrb{e^{- \nu \binom \nalt2}, 1}, 
\end{equation}
where $n_d \coloneqq \tfrac14 \sum_{\alt, \alt' \in [\nalt]} \lra{s_1(\alt, \alt') - s_2(\alt, \alt')}$ is the number of discordant pairs and
$s_l(\alt, \alt') \coloneqq \sign\lrp{\ra_l(\alt)-\ra_l(\alt')}$ with $\sign(0) = 0$.

The following assumption is central to the derivations in the rest of this section.
\begin{assumption}[Interpretability]
    \label{ass:interpretability}
    $k$ factorizes as $k = g \circ d$, where $d: \m X \times \m X \to [0, 1]$ is a normalized dissimilarity\footnote{A dissimilarity is a function that is symmetric and such that $d(x, x) \equiv 0$; it is normalized if it takes values in $[0, 1]$.} and $g: [0, 1] \to [0, 1]$ is a monotonically decreasing convex function with $g(0) = 1$.
\end{assumption}

\begin{remark}\label{remark:interpretability}
    The kernels of \cref{sec:externalvalidity:kernels:kernelsforrankings} satisfy \cref{ass:interpretability} for every value of their parameters, and so does the RBF kernel on $[0,1]^{n_a}$ (\cref{app:kernels:vector}), with $d$ and $g$ as in \cref{tab:interpretable_kernels}. 
\end{remark}
    \begin{table}[t]
    \caption{Interpretable form $k = g\circ d$ (\cref{ass:interpretability}) of the kernels used in this work.}
    \label{tab:interpretable_kernels}
    \begin{center}
    \begin{tabular}{lll}
    \toprule
    Kernel & $d(x_1,x_2)$ & $g(u)$ \\
    \midrule
    Borda $k^{a^*,\nu}_\mathrm{b}$ & $|b_1-b_2|/(n_a-1)$ & $e^{-\nu(n_a-1)u}$ \\
    Jaccard $k^t_\mathrm{j}$ & $1-J_t(r_1,r_2)$ & $1-u$ \\
    Mallows $k^\nu_\mathrm{m}$ & $n_d/\binom{n_a}{2}$ & $e^{-\nu\binom{n_a}{2}u}$ \\
    RBF $k^\gamma_\mathrm{RBF}$  & $\|x_1-x_2\|^2/n_a$ & $e^{-\gamma n_a u}$ \\
    \bottomrule
    \end{tabular}
    \end{center}
\end{table}

For a desired maximum expected dissimilarity $\delta \in [0, 1]$, we can define $\eps \in \R^+$ to obtain an upper bound on the expected $\MMD$, as stated in the following lemma, proved in \cref{app:proofs:epsfromdelta}.
\begin{lemma}[Choice of $\eps$]
\label{lemma:epsfromdelta}
    Under \cref{ass:interpretability}, let $\delta \in [0, 1]$ and $\eps \coloneqq \sqrt{2 \lrp{1 - g(\delta)}}$.
    Then, for any $n \geq 1$,
    \begin{align}
        \E_{P \otimes P} d(X, X') \leq \delta
            &\implies \E_{P^{2n}} \MMD_n \leq \eps .
    \end{align}
\end{lemma}
We next detail how the KEV helps assess the appropriate size of a study.

\subsection{On the Size of a Valid Study}
\label{sec:externalvalidity:sizeofstudy}

\cref{lemma:KEVmonotonicity} suggests that the KEV is overall increasing in the sample size $n$.
Thus, we define the minimum number of samples for a study to be $(\alpha, \eps, n)$-valid as
\begin{align}
\label{eq:nstardefinition}
    n^* \coloneqq \min \lrc{n: \ev_n(P, \eps) \geq \alpha}.
\end{align}
Under \cref{ass:interpretability}, the KEV is lower-bounded by a function that is monotonically increasing in all of $\eps$, $n$, and the $P\otimes P$-expectation of $k$, as detailed in the following lemma; its proof is in \cref{app:proofs:KEVmarkovlowerbound}.
\begin{lemma}[Lower bound]
\label{lemma:KEVmarkovlowerbound}
    For any kernel $k$ in the form of \cref{ass:interpretability}, $P \in \m M_1^+(\m X)$ and $\eps \in \R^+$, 
    \begin{align}
        \ev_n(P, \eps) \geq  1- \frac{2}{\eps^2 n}(1- \bar k_P),
    \end{align}
    where $\bar k_P \coloneqq \E_{P \otimes P} k(X, X')$.
\end{lemma}
As a consequence, 
\begin{align}
\label{eq:nstarupperbound}
    n^* \leq \frac{2\lrp{1-\bar k_P}}{(1-\alpha)\eps^2},
\end{align}
and any study with results $P$ is $(\alpha, \eps, n)$-valid for every $n$ at least the right-hand side of \eqref{eq:nstarupperbound}.

Bound \eqref{eq:nstarupperbound} is rather loose: for a tighter estimate of $n^*$ that is viable in practice, we derive the following approximate power law between $n$ and the quantiles of $\MMD_n$, which we prove in \cref{app:proofs:mmdapproxpowerlaw}.
\begin{theorem}[Power Law]
\label{thm:mmdapproxpowerlaw}
    Let $k$ be a bounded and non-constant kernel, $\alpha \in [0, 1]$, and $q_\alpha(n) \coloneqq F_{\MMD_n}^{-1}(\alpha)$.
    Then, there exists $\beta_\alpha \in \R$ such that
    \begin{align}
    \label{eq:mmdpowerlaw}
        \log q_\alpha(n) + \frac12 \log n = \beta_\alpha + o(1).
    \end{align}
\end{theorem}
\begin{remark}
\label{remark:powerlaw_getnstar}
    Given a kernel $k$ satisfying the assumptions of \cref{thm:mmdapproxpowerlaw} and desired $\alpha\in[0, 1]$ and $\eps\in\R^+$, we can fit $\beta_\alpha$ as in \eqref{eq:mmdpowerlaw} and estimate $\log n^* = 2\beta_\alpha -2\log \eps$. 
    If instead $k(x, x')\equiv c$, we have that $q_\alpha(n) = 0$ for any $\alpha$ and $n$, and we set $n^* = 1$.
\end{remark}

To discuss practical applications of \cref{thm:mmdapproxpowerlaw} and \cref{remark:powerlaw_getnstar}, we first provide the following consistency result for the plug-in bootstrap estimator of the KEV, proved in \cref{app:proofs:bootstrap_consistency}.
\begin{lemma}[Bootstrap Consistency]
\label{lemma:bootstrap_consistency}
    Let $\b Y \sim P^N$ and $P_N \coloneqq \hat P_{\b Y}$ be its empirical distribution.
    Define $\widehat \MMD_{n; N}|\b Y \coloneqq \MMD(\hat P_{\b Z}, \hat P_{\b Z'})$, where $\b Z, \b Z' \sim P_N^n$ independently. 
    
    Then, for $P^N\aew\ \b Y \in \m X^N$ and for $n, N \to \infty$, $n\MMD_n^2$ and $\tfrac N2 \widehat \MMD^2_{N/2; N}|\b Y$ converge in distribution to the same random variable.
\end{lemma}

With this result, we are ready to propose incorporating the estimate of $n^*$ provided by \cref{thm:mmdapproxpowerlaw} into the execution of an experimental study, as described in \cref{alg:studyexecution}. 
\begin{algorithm}[t]
    \caption{Incremental execution of a study}
    \label{alg:studyexecution}
    \begin{algorithmic}[1]
        \Require kernel $k$, validity level $\alpha$, dissimilarity $\delta$, increment $N_0$
        \State $\varepsilon \gets \varepsilon(\delta)$; $\b y \gets \emptyset$; $N \gets 0$; $\hat n^*_N \gets \infty$
        \While{$\hat{n}^*_{N} > N$}
            \State $\b y \gets \b y \,\|\, \texttt{RunExperiments}(N_0)$        
            \State $N \gets \lra{\b y}$
            \State $\hat{n}^*_N \gets \texttt{Estimate}n^*(\b y, k, \alpha, \eps)$ \Comment{cf.\ \cref{remark:powerlaw_getnstar}}
        \EndWhile\\
        \Return $\b y$
    \end{algorithmic}
\end{algorithm}
Note that stopping at the first step where $\hat n^*_N \leq N$ may end the study prematurely if $\hat n^*_N$ underestimates $n^*$.
A possible mitigation, which we do not consider in this work, would be to add a safety margin $\rho >0$ and to stop only when $(1 + \rho)\hat n^*_N \leq N$.

Finally, \cref{app:resampling:schemes} compares four resampling schemes, including the bootstrap covered by \cref{lemma:bootstrap_consistency} and a permutation scheme which mimics actual replications, as it never repeats an experimental condition; \cref{app:resampling:incremental} compares incremental and fresh sampling of $\b y$ in \cref{alg:studyexecution}.

\section{EXPERIMENTS}
\label{sec:experiments}

This section first compares the bootstrap and permutation approaches to resampling. 
Then, we show how to analyze experimental studies for their external validity, treating every configuration as independent. 
Finally, we show a practical application of \cref{alg:studyexecution} to a simulated experimental study.

In our experiments, we set $\nu = 1/(n_a-1)$ for $\kernel[b]$, $\nu = 1/\binom{n_a}{2}$ for $\kernel[m]$, and $\gamma = 1/n_a$ for $\kernel[RBF]$, so that $\varepsilon(\delta) = \sqrt{2(1-e^{-\delta})}$ for these three kernels; for $\kernel[j]$, for which we set $t=1$, we have $\varepsilon(\delta) = \sqrt{2\delta}$.

For both case studies, we use the bootstrap estimator.

\textbf{Runtime}\; 
The main bottleneck of the estimation is bootstrapping the distribution of $\MMD_n$. 
Our implementation has complexity $\m O(n_u^2 n_{\text{rep}})$, where $n_u$ is the number of unique results obtained and $n_{\text{rep}}$ is the number of resamples. 
On a consumer laptop, the total runtime\footnote{For 48 configurations, 4 kernels, for every $n= 1, \dots, 20$, and for 200 resamples.} was $7$ minutes in the encoders case study (\cref{sec:case_studies:encoders}), averaging about $2.2$ seconds per configuration-kernel pair.
The runtime is comparable for the LLM case study of \cref{sec:case_studies:LLMs}.

\subsection{Resampling Methods}
\label{sec:experiments:estimator}

\begin{figure}[t]
    \centering
    \includegraphics[width=0.95\linewidth]{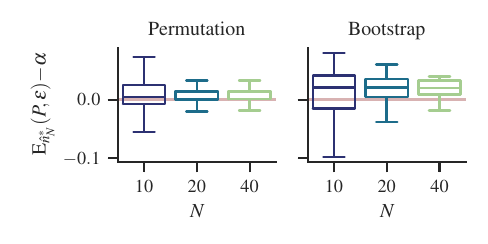}
    \caption{Actual ($\ev_{\hat n^*_N}(P, \eps)$) vs. desired $(\alpha)$ validity.}
    \label{fig:alpha_vs_evnstar}
\end{figure}

We first investigate the permutation and bootstrap approaches to estimate the minimum number of experiments needed to achieve $(\alpha=0.95, \eps(\delta=0.05), n^*)$-valid results for the Mallows kernel as in \eqref{eq:mallows_kernel}. 
In particular, we compare the desired validity, $0.95$, with the actual validity at the estimated $n^*_N$, $\ev_{\hat n^*_N}(P, \eps)$: the ideal approach is cautiously pessimistic and shows a slightly positive difference between the two.
We use a sample of size $N=10, 20, 40$ from $P = \m U(\m R_8)$,\footnote{The uniform distribution over rankings of 8 alternatives; $\lra{\m R_8} = 545\,835$, cf. \url{https://oeis.org/A000670}.} and perform both a bootstrap and a permutation approach to (1) get the distribution of $\MMD_n$ for $n = 1, \dots, N/2$; (2) get their quantiles $q_\alpha(n)$; (3) fit $\log n = -2 \log q_\alpha(n) + 2\beta_\alpha$ according to \cref{thm:mmdapproxpowerlaw}; (4) define $\log \hat n^*_N = 2\beta_\alpha -2 \log \eps$; (5) compute $\ev_{\hat n^*_N}(P, \eps)$ and compare it with $\alpha$.
\cref{fig:alpha_vs_evnstar} shows that both methods yield an actual validity within $0.1$ of $\alpha$ already for $N=10$, and within $0.05$ for $N\ge20$: reasonable estimates of $n^*$ can thus be obtained even from small preliminary studies. \Cref{app:resampling} contains additional details.

\subsection{Case Study: Encoders}
\label{sec:case_studies:encoders}

\begin{figure}[t]
    \centering
    \includegraphics[width=0.99\linewidth]{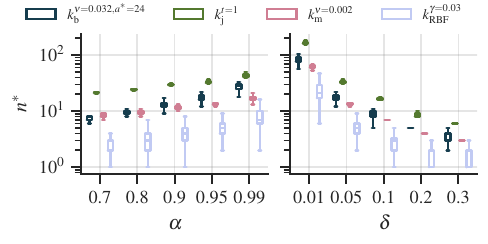}
    \caption{Necessary size $n^*$ to achieve $(\alpha, \eps(\delta), n^*)$-validity. (Left) $\delta = 0.05$; (Right) $\alpha = 0.95$. Variability across configuration.}
    \label{fig:encoders_nad}
\end{figure}
    
\begin{figure}
    \centering
    \includegraphics[width=0.93\linewidth]{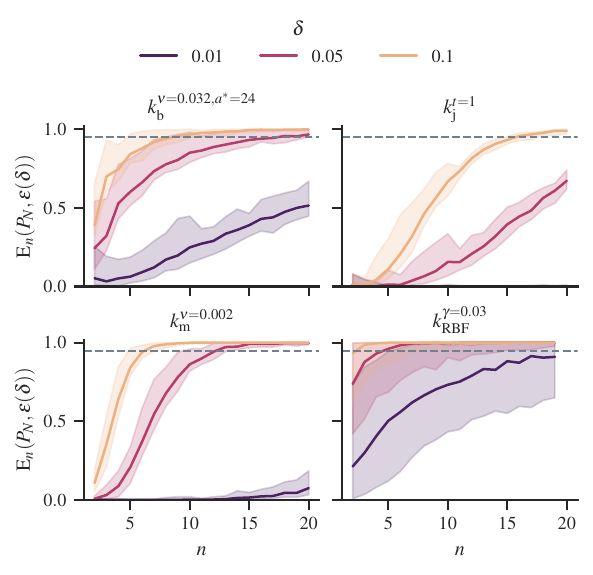}
    \caption{KEV for varying $\delta$-s and $n$-s. The shaded regions correspond to all configurations.}
    \label{fig:encoders_rn}
\end{figure}

We now assess the external validity of \citet{matteucci2023encoders}, an experimental study evaluating encoders for categorical data.\footnote{Experimental results released with the original study; MIT license; accessed on the 5th of October 2026.} 
The study evaluates 32 encoders on binary classification datasets using as proxy the performance of a model trained on the encoded data.
We consider as fixed factors \texttt{model}, \texttt{tuning}, and \texttt{metric}; as random factor \texttt{dataset}.
We are interested in assessing the external validity when datasets are sampled at random from the pool of datasets considered in the original study, with the other factors fixed.

We include in the analysis only the experimental conditions in which $\ge80\%$ of the encoders were run and, among those, the encoders that were run on $\ge80\%$ of the remaining conditions.\footnote{Remaining missing evaluations are imputed with the worst rank.}
This leaves us with $32 / 32$  alternatives, $1872 / 2036$ conditions, and $48$ configurations.
We consider the following research questions (kernels):
\textit{``Does \textnormal{\texttt{one-hot}} ($a^*=24$) rank consistently?''} ($\kernel[b]$);
\textit{``Are the best encoders consistently the same ones?''} ($\kernel[j]$);
\textit{``Are the encoders ranked consistently?''} ($\kernel[m]$);
\textit{``Are the raw performances of the encoders consistent?''} ($\kernel[RBF]$).
To show how the framework behaves on numerical data, we include the radial basis function (RBF) kernel; further details are in \cref{app:kernels:vector}.

\cref{fig:encoders_nad} shows that the number of experiments needed to achieve $(\alpha, \eps(\delta), n^*)$-validity is increasing with $\alpha$ and decreasing with $\delta$.
\cref{fig:encoders_rn} shows that the validity depends on the kernel chosen and increases with both sample size $n$ and $\delta$; recall that a greater $\delta$ corresponds to a looser condition on what results we consider similar, cf. \cref{lemma:epsfromdelta}.
The two figures agree: taking, for instance, $\kernel[j]$, \cref{fig:encoders_nad} shows that the median configuration achieves $(0.95, \eps(0.1), \approx15)$-validity, and \cref{fig:encoders_rn} shows that the results are indeed valid for any $n \gtrsim 15$.

\subsection{Case Study: LLMs}
\label{sec:case_studies:LLMs}

\begin{figure}[htb]
    \centering
    \includegraphics[width=0.99\linewidth]{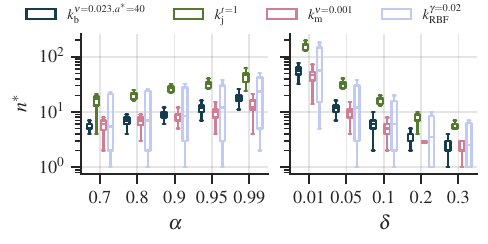}
    \caption{Necessary size $n^*$ to achieve $(\alpha, \eps(\delta), n^*)$-validity. (Left) $\delta = 0.05$; (Right) $\alpha = 0.95$. Variability across configuration.}
    \label{fig:llms_nad}
\end{figure}
    
\begin{figure}[htb]
    \centering
    \includegraphics[width=0.93\linewidth]{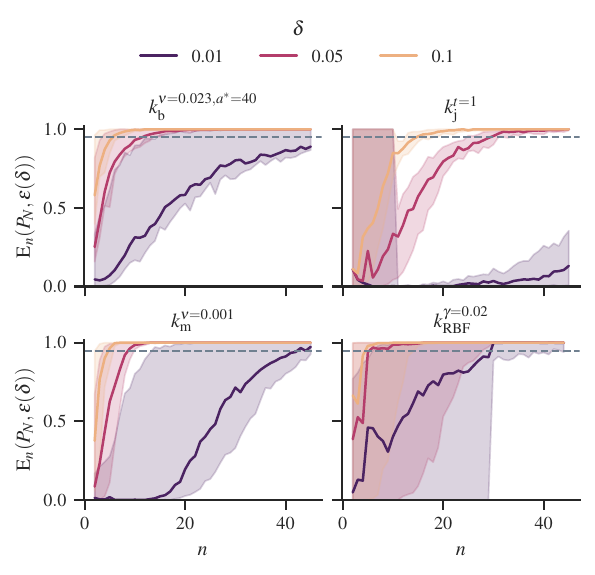}
    \caption{KEV for varying $\delta$-s and $n$-s. The shaded regions correspond to all configurations.}
    \label{fig:llms_rn}
\end{figure}

We now turn to the other case study, \citet{srivastava2023beyond}, a collaborative benchmark of Large Language Models (LLMs).\footnote{Experimental results released with the original study; MIT license; accessed on the 5th of October 2026; Apache 2.0 license.} 
The study evaluates 55 LLMs on a suite of tasks, each with its subtasks and a recommended evaluation metric.
We take \texttt{task} and \texttt{num\_shots} as fixed factors, and \texttt{subtask} as a random factor.
We are therefore interested in assessing the external validity when subtasks are sampled at random from the pool of those considered in the original study, with the other factors fixed.

We apply the same filtering as in \cref{sec:case_studies:encoders}, which leaves $45 / 55$  alternatives and $2269 / 5448$ conditions across $24$ configurations.

We consider the following research questions (kernels):
\textit{``Does \textnormal{\texttt{GPT-3}} ($a^*=40$) rank consistently?''} ($\kernel[b]$);
\textit{``Are the best LLMs consistently the same ones?''} ($\kernel[j]$);
\textit{``Are the LLMs ranked consistently?''} ($\kernel[m]$);
\textit{``Are the raw performances of the LLMs consistent?''} ($\kernel[RBF]$).

\cref{fig:llms_nad} shows that the size needed to achieve $(\alpha, \eps(\delta), n^*)$-validity is increasing with $\alpha$ and decreasing with $\delta$.
\cref{fig:llms_rn} illustrates how the KEV depends on the kernel chosen and increases with both sample size $n$ and $\delta$; recall that a greater $\delta$ corresponds to a looser condition on what results we consider similar, cf. \cref{lemma:epsfromdelta}.
Compared to the case study of \cref{sec:case_studies:encoders}, there are more pathological examples where the KEV is close to 1 for any $n$.
This can arise when most LLMs fail most subtasks of a task: their rankings are then tied, and coincide regardless of the subtask.

\subsection{Proposed Execution of a Study}
\label{sec:experiments:execution}

\begin{figure*}[t!]
    \centering
    \includegraphics[]{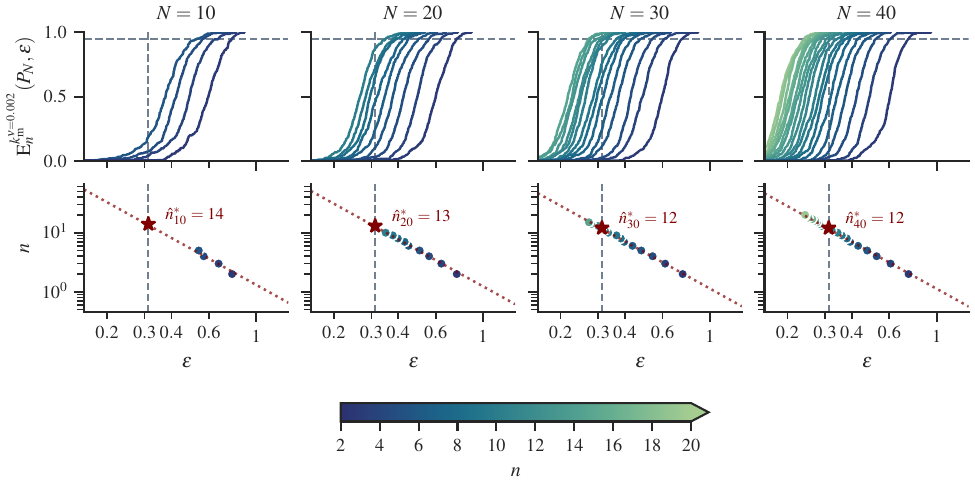}
    \caption{Proposed execution of an experimental study. (Top) The KEV is estimated from $N$ preliminary experiments, for $n\in [\lfloor N/2 \rfloor]$. (Bottom) A power law relation is fitted between $n$ and the $0.95$-quantile of $\MMD_n$; $\hat n^*_N$ is the prediction at $\eps(\delta=0.05) = 0.312$ (\textcolor{Mahogany}{$\star$}).} 
    \label{fig:encoders_nstarestimation}
\end{figure*}

In this section, we simulate\footnote{The results are precomputed and sampled in batches.} running studies following \cref{alg:studyexecution} to determine a reasonable number of experiments $n^*$ yielding $(\alpha, \eps, n^*)$-valid results, see \eqref{eq:nstardefinition}.
We apply it to the fixed configuration $(\texttt{model},\texttt{tuning}, \texttt{metric}) = (\texttt{DTC}, \texttt{no}, \texttt{F1})$ from~\citet{matteucci2023encoders}, the Mallows kernel \eqref{eq:mallows_kernel}, $\alpha=0.95$, and $\eps(\delta = 0.05) = 0.312$ (cf. \cref{lemma:epsfromdelta}), and we run experiments in increments of $N_0=10$, as depicted in \cref{fig:encoders_nstarestimation}.

From $N$ experiments, we estimate the KEV and, from it, the $\alpha$-quantiles of $\widehat{\mathrm{MMD}}_{n;N}|\mathbf y$ for $n = 1, \dots, \lfloor N/2 \rfloor$ (top row); we then fit the power law of Theorem~\ref{thm:mmdapproxpowerlaw} and compute $\hat n^*_N$ as in Remark~\ref{remark:powerlaw_getnstar} (bottom row). 
If $\hat n^*_N > N$, we sample $N_0$ more datasets, run the corresponding experiments, and repeat the analysis (left to right).
\Cref{alg:studyexecution} stops at $N=20$, where $\hat n^*_{20}=13\leq N$, and the estimates obtained from more experiments are stable: $\hat n^*_{30}=\hat n^*_{40}=12$. 
The prediction agrees with the KEV estimated from all $N=40$ available results, which include the preliminary ones: $\ev_{13}(P_{40},\varepsilon)=0.975\geq\alpha$. 

\section{LIMITATIONS}
\label{sec:conclusions}

Our framework requires a factorial design in which the levels of the random factors can be sampled independently; as a consequence, we do not cover, for instance, active experimental designs.
Our estimates concern the pool of experimental conditions considered in a study, that is, $P_N$; extending them to unseen conditions, such as new datasets, requires the pool to be representative of the population of interest. 
Moreover, we impute missing evaluations with the worst rank, whose effect on the KEV we do not quantify, and we validate the estimation of $n^*$ on a single synthetic distribution.
Finally, in our implementation, we did not consider more efficient computational methods for estimating the distribution of the MMD, such Nystr\"om approaches \citep{chatalic2026nystroemmmd}.

\section*{AI use statement}
In this work, we used generative AI tools for mathematical proofs and method implementation. 
We have not used generative AI tools for synthetic-data generation, theoretical or conceptual development, hypotheses, research methodology, experiments, or interpretation of results.
Translation, dataset cleaning or reformatting, and qualitative analysis are not applicable to this work. 
Additionally, we used generative AI tools for literature discovery, coding, and drafting text. 
We have reviewed all AI-assisted work. 
The authors read and verified the validity of the literature found. 
LLM-assisted proofs were checked line by line by the first author. 
LLM-produced code was thoroughly checked and tested and mainly concerns optimizations of pre-existing code. 
We take responsibility for the final content of this work, including text, claims, or artifacts produced with the aid of generative AI.

\bibliography{bibarxiv}

\clearpage
\appendix
\thispagestyle{empty}

\onecolumn
\aistatstitle{Kernel-based External Validity of Experimental Studies: \\
Supplementary Materials}
 
\appendix
\crefalias{section}{appendix}
\crefalias{subsection}{appendix}

\paragraph{Outline}
\cref{app:studies} extends the classification of experimental factors from \cref{sec:background:studies};
\cref{app:kernels} discusses in detail the kernels of \cref{sec:externalvalidity:kernels:kernelsforrankings}, including the related proofs;
\cref{app:resampling} evaluates different resampling schemes;
\cref{app:proofs} contains the proofs of the results in the main text as well as of additional results.

\section{Experimental factors (extended)}
\label{app:studies}

This section elaborates on the classification of (experimental) factors in~\cref{sec:background:studies}.
We recall that a factor is anything that may affect the result of an experiment; we sort factors into categories adapted from~\citet[Chapter~1]{montgomery2017design}, according to the role that a factor plays in a study.
\begin{itemize} 
    \item \emph{Noise factors} are included in the study only to ``average out'' the effect of randomness in either the benchmarked algorithms or the experimental conditions; the results are usually aggregated w.r.t. these factors before proceeding to further analysis (for instance, by taking the average); common examples are \texttt{seed} or \texttt{cross-validation fold}.
    \item \emph{Held-constant factors} are presumed not to significantly impact the results; they are thus out of the study's scope and are fixed to a single level; examples might be \texttt{programming language} or \texttt{operating system}. 
    \item \emph{Design factors} are expected to significantly impact the results and have a relatively small set of levels; 
    examples are \texttt{metric} or \texttt{number of shots}. 
    \item \emph{Random factors}
    have a larger number of levels; the experimenter's goal is to obtain results that are stable under changes of these factors' levels, that is, to perform a study which is random-factor-externally-valid; examples are \texttt{dataset} or \texttt{chessboard position}. 
\end{itemize}
Notice that the classification above is purely semantic: the same factor may play different roles in different studies, according to the study's objectives.
For instance, a large-scale comparison study of classifiers will likely consider \texttt{dataset} as a random factor, as the results are expected to generalize w.r.t. the choice of dataset. 
Conversely, a study focusing on evaluating prediction models on a single dataset will consider \texttt{dataset} as a held-constant factor and \texttt{resample} as a random factor.

\section{Kernels}
\label{app:kernels}

\subsection{Kernels for Rankings}
\label{app:kernels:rankings}

This section proves that the functions $\Ra \times \Ra \to \R$ introduced in \cref{sec:externalvalidity:kernels:kernelsforrankings} are indeed kernels for rankings.
In other words, we prove that they are symmetric and positive definite.
As all three kernels are clearly symmetric, it remains to show their positive definiteness.

We now proceed with the proofs.
\begin{lemmaA}[The Borda kernel is positive definite]
    Let, for $a^* \in [\nalt]$ and $\nu \in \R^+$,
    \begin{equation}
    \label{app:eq:borda_kernel}
        \kernel[b][\alt^*, \nu]\lrp{\ra_1, \ra_2} =  e^{-\nu \lra{b_1 - b_2}} \quad \in \lrb{e^{-\nu (\nalt-1)}, 1},
    \end{equation}
    where $b_l = \lra{\lrc{\alt\in[\nalt]: \ra_l(\alt) \geq \ra_l(\alt^*)}}$ is the Borda count~\citep{borda1781memoire}.
    Then, $\kernel[b]$ is positive definite.
\end{lemmaA}
\begin{proof}
    We show that there exists a distance $d$ such that $(\Ra, d)$ embeds isometrically into some $L^2$ space.
    This ensures that $-d$ is a conditionally positive definite (c.p.d.) function and, thus, that $e^{-\nu d}$ is positive definite~\citep{schoenberg1938metric, schoelkopf2000kernel}.
    Define the distance 
    \begin{align}
        d : \Ra \times \Ra &\to \R^+\\
            (r_1, r_2)  &\mapsto \lra{b_1-b_2},
    \end{align}
    where $b_l$ is as in \eqref{app:eq:borda_kernel}.
    Now, $(\Ra, d)$ is isometrically embedded into $(\R, \lrn{\cdot}_2)$ via the map $r_l \mapsto b_l$. 
    Hence, $-d$ is c.p.d. and $\kernel[b]$ is positive definite.
\end{proof}

\begin{lemmaA}[The Jaccard kernel is positive definite]
    Let, for a tier $t \in \N$, 
    \begin{equation}
        \label{app:eq:jaccard_kernel}
        \kernel[j][t]\lrp{\ra_1, \ra_2} \coloneqq J_t(\ra_1, \ra_2) = \frac {\lra{\ra_{1, [t]} \cap \ra_{2, [t]}}} {\lra{\ra_{1, [t]} \cup \ra_{2, [t]}}} \quad\in [0, 1],
    \end{equation}
    where, for $l=1, 2$, $\ra_{l, [t]} \coloneqq \lrc{\alt\in[\nalt]: \ra_l(\alt) \leq t}$ is the set of alternatives with rank at most $t$.
    Then, $\kernel[j]$ is a kernel.
\end{lemmaA}
\begin{proof}
    It is already known that the Jaccard coefficient is a kernel between sets, see e.g. \citet{bouchard2013jaccard}.
    As $\kernel[j]$ is the Jaccard coefficient applied to the top-$t$ tiers of two rankings, it is also positive definite.
\end{proof}

\begin{lemmaA}[The Mallows kernel is positive definite]
    Let, for $\nu \in \R^+$,
    \begin{equation}
    \label{app:eq:mallows_kernel}
        \kernel[m][\nu](\ra_1, \ra_2) = e^{-\nu n_d} \quad \in \lrb{e^{- \nu \binom \nalt2}, 1}, 
    \end{equation}
    where $n_d \coloneqq \tfrac14 \sum_{\alt, \alt' \in [\nalt]} \lra{s_1(\alt, \alt') - s_2(\alt, \alt')}$ is the number of discordant pairs and
    $s_l(\alt, \alt') \coloneqq \sign\lrp{\ra_l(\alt)-\ra_l(\alt')}$ with $\sign(0) = 0$.
    Then, $\kernel[m]$ is positive definite.
\end{lemmaA}
\begin{proof}
    It is known that the number of discordant pairs $n_d$ is a distance on $\Ra$~\citep{kemeny1962mathematical}.
    Consider now the mapping of a ranking into its adjacency matrix,
    \begin{align}
        \Phi: \Ra &\to \lrc{0,1}^{\nalt\times\nalt} \\
                r &\mapsto \lrp{\ind\lrp{r(i) \leq r(j)}}_{i, j=1}^\nalt,
    \end{align}
    where $\ind$ is the indicator function.
    Then, 
    \[
        2n_d = \lrn{\Phi(r_1)-\Phi(r_2)}_1 = \lrn{\Phi(r_1)-\Phi(r_2)}_2^2 
    \] 
    where $\lrn{\cdot}_p$ indicates the entry-wise matrix $p$-norm and the equality holds because the entries of the matrices are either 0 or 1.
    Hence, $\kernel[m]^\nu$ is a Gaussian (RBF) kernel on $\R^{\nalt\times \nalt}$ with parameter $\gamma = \nu/2$, evaluated at $(\Phi(r_1),\Phi(r_2))$. 
    As the Gaussian kernel is positive definite and positive definiteness is preserved under composition with any map, $\kernel[m]$ is itself positive definite. 
\end{proof}

\subsection{Kernels for Euclidean Spaces}
\label{app:kernels:vector}

In the experiments of \cref{sec:experiments}, we also consider Euclidean results and the radial basis function (RBF; Gaussian) kernel~\citep{schoelkopf2022kernels},
using the scikit-learn implementation~\citep{scikit-learn}.
The \textit{RBF kernel} $\kernel[RBF]$ is suitable for research questions of the form \emph{``Are the performances consistent?''}
For $\gamma \in \R^+$ and $\b x_1, \b x_2 \in \R^\nalt$,
\begin{equation}
\label{app:eq:rbf_kernel}
    \kernel[RBF]^\gamma(\b x_1, \b x_2) = e^{-\gamma \lrn{\b x_1 - \b x_2}^2}, 
\end{equation}
where $\lrn{\cdot}$ is the Euclidean norm. 
Notice that, for results in $[0,1]^{n_a}$, such as the scores in our case studies, $\kernel[RBF]^\gamma$ satisfies \cref{ass:interpretability} with $d(x_1,x_2) = \|x_1-x_2\|^2/n_a$ and $g(u) = e^{-\gamma n_a u}$; with $\gamma = 1/n_a$, as in our experiments, $g(u) = e^{-u}$.

\section{Notes on resampling}
\label{app:resampling}

This section empirically evaluates different estimators and resampling schemes to determine the minimum number of experiments needed to achieve valid results.

\subsection{Resampling schemes}
\label{app:resampling:schemes}

\begin{figure*}[t!]
    \begin{subfigure}[t]{0.5\textwidth}
        \centering
        \includegraphics[width=0.95\linewidth]{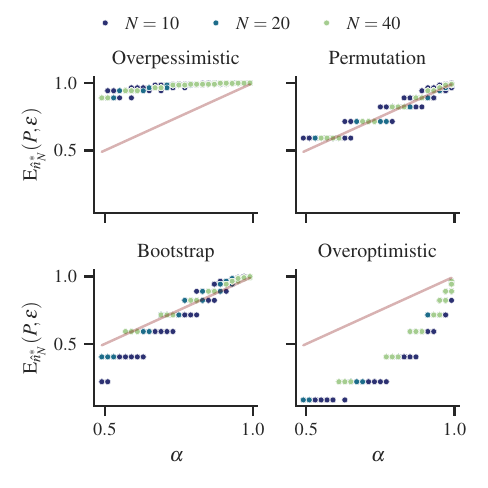}
        \caption{Comparison plots}
    \end{subfigure}
    ~
    \begin{subfigure}[t]{0.5\textwidth}
        \centering
        \includegraphics[width=0.95\linewidth]{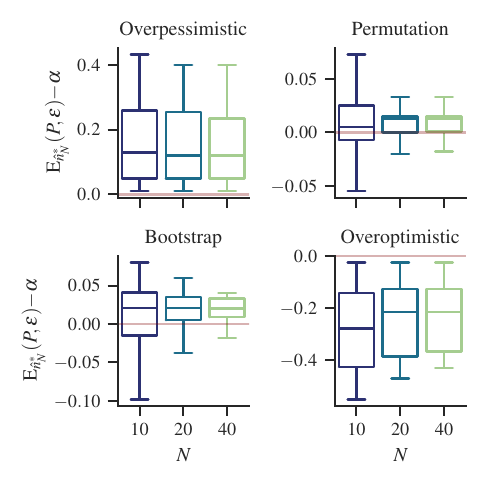}
        \caption{Boxplots}
    \end{subfigure}
    \caption{KEV at the estimated $\hat n^*_N$ compared to the desired level $\alpha$. 
    The cautiously pessimistic scenario corresponds to a slightly positive difference.
    $P=\m U(\m R_8)$ and we use the Mallows kernel.}
    \label{fig:mmdresampling}
\end{figure*}

In this section, we expand on \cref{sec:experiments:estimator} and provide additional details on the design of our implementation.
Estimating $n^*$ requires resampling pairs of samples $\b {x}, \b {x}'$ of size $n$ from the $N$ available experimental results $\b {y}$, with distribution $P_N$.
\cref{fig:mmdresampling} expands \cref{fig:alpha_vs_evnstar} and compares four ways of doing so, defined by two flags:
\texttt{disjoint} controls whether $\b {x}$ and $\b {x}'$ are drawn from two disjoint halves of $\b {y}$, so that no experiment appears in both (we still allow the same value to appear in both);
\texttt{replace} determines whether $\b {x}$ and $\b {x}'$ are drawn with replacement.

Notice that, by \eqref{eq:mmd_empirical}, $\MMD_n^2$ is the sum of two within-sample terms minus a between-sample term: collisions within a sample bias $\MMD_n$ upwards; collisions between the two samples bias it downwards.
The scheme with \texttt{disjoint} = False and \texttt{replace} = True is the one covered by \cref{lemma:bootstrap_consistency}.

\begin{table}[thb]
    \centering
    \caption{Resampling schemes and their expected effect on the estimated KEV.}
    \label{tab:mmdresampling}
    \begin{tabular}{llccl}
    \toprule
    \texttt{disjoint} & \texttt{replace} & Collisions within & Collisions between & Scheme \\
    \midrule
    True  & True  & yes & no  & overpessimistic \\
    True  & False & no  & no  & permutation \\
    False & False & no  & yes & overoptimistic \\
    False & True  & yes & yes & bootstrap \\
    \bottomrule
    \end{tabular}
\end{table}

\subsection{Incremental vs. Fresh Sampling}
\label{app:resampling:incremental}

\cref{alg:studyexecution} builds the preliminary study incrementally in batches of $N_0$ experiments, appending the new results to the existing ones, which is the practically relevant scenario.
The results of \cref{sec:externalvalidity:sizeofstudy}, instead, correspond to fresh sampling: at every step, a new sample $\b {y} \sim P^{N}$ is drawn independently of the previous ones.
As the schemes differ in the joint law of the sequence of samples, we empirically compare the behavior of the KEV for the two case studies of \cref{sec:experiments}.

\cref{fig:mmdresampling_fresh_incremental} shows that the curves nearly coincide for every kernel; the only visible deviation, for $k_\mathrm{j}$ in the LLM case study at small $n$, occurs where the KEV is far below $\alpha$ and hence does not affect $\hat{n}^*_N$.
In practice, the incremental scheme can therefore be used in place of the fresh one.

\begin{figure*}[thb]
    \begin{subfigure}[t]{0.5\textwidth}
        \centering
        \includegraphics[width=0.95\linewidth]{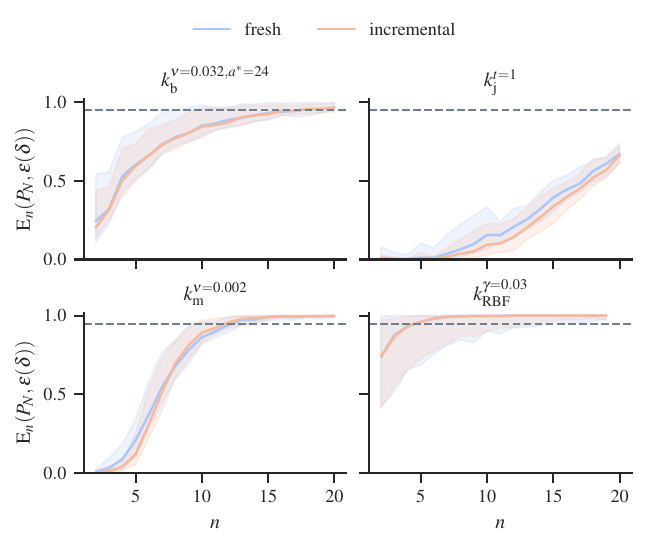}
        \caption{Encoders}
    \end{subfigure}
    ~
    \begin{subfigure}[t]{0.5\textwidth}
        \centering
        \includegraphics[width=0.95\linewidth]{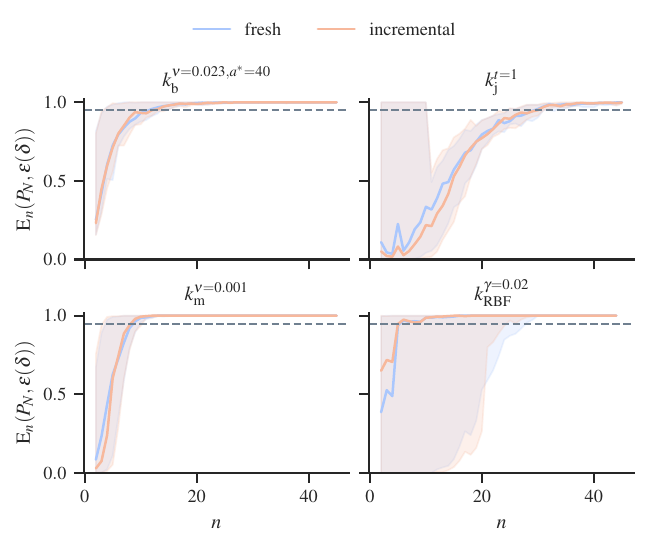}
        \caption{LLMs}
    \end{subfigure}
    \caption{KEV compared for incremental and fresh sampling schemes.}
    \label{fig:mmdresampling_fresh_incremental}
\end{figure*}

\section{Proofs}
\label{app:proofs}

We first state and prove auxiliary results not in the main text.

\subsection{Auxiliary Results}
\label{app:proofs:additionalresults}

\begin{lemmaA}
\label{lemma:ksupinterpretable}
    Let $k = g \circ d$ as in \cref{ass:interpretability}. Then, $\kernel(x,x) \equiv 1$.
\end{lemmaA}
\begin{proof}
    \begin{align}
        \kernel(x, x) = g(d(x, x)) = g(0) = 1.
    \end{align}
\end{proof}

\begin{lemmaA}[$\E \MMD^2_n$]
\label{lemma:EMMDn}
    Let $\MMD_n$ be as in \eqref{eq:mmdn_def}, and let $k$ be a kernel with constant diagonal $k(x, x) \equiv \ksup$.
    Then, for any $n \geq 1$,
    \begin{align}
        \E_{P^{2n}}\MMD_n^2 = \frac2n \lrp{\ksup - \bar k_P},
    \end{align}
    where $\bar k_P \coloneqq \E_{P \otimes P} k(X, X')$.
\end{lemmaA}
\begin{proof}
    \begin{align}
    \label{eq:app:EMMDn_expanded}
        \E_{P^n  \otimes P^n}\MMD_n^2 
            &= \E_{P^n} \lrang{\mu_{\hat P_{\b X}}, \mu_{\hat P_{\b X}}}_{\m H} + \E_{P^n} \lrang{\mu_{\hat P_{\b X'}}, \mu_{\hat P_{\b X'}}}_{\m H} - 2\E_{P^n  \otimes P^n} \lrang{\mu_{\hat P_{\b X}}, \mu_{\hat P_{\b X'}}}_{\m H}
    \end{align}
    Now, by linearity of the mean embedding, $\mu_{\hat P_{\b X}} = \tfrac1n \sum\limits_{i=1}^n \kernel(X_i, \cdot)$ and, by $\kernel$ being the reproducing kernel for $\m H$, 
    \begin{align}
        \lrang{\mu_{\hat P_{\b X}}, \mu_{\hat P_{\b X}}}_{\m H} 
            &= \frac1{n^2} \sum\limits_{i, j = 1}^n \kernel(X_i, X_j).
    \end{align}
    Taking expectations,
    \begin{align}
        \E_{P^n} \lrang{\mu_{\hat P_{\b X}}, \mu_{\hat P_{\b X}}}_{\m H}
            &= \frac1{n^2} \sum\limits_{i, j = 1}^n \E_{P\otimes P} \kernel(X_i, X_j) \\
            &\overset{(a)}{=} \frac1{n^2} \sum\limits_{i = 1}^n \E_{P\otimes P} \kernel(X_i, X_i) +  \frac1{n^2} \sum\limits_{i\neq j = 1}^n \E_{P\otimes P} \kernel(X_i, X_j)\\
            &\overset{(b)}{=} \frac{1}{n} \ksup +  \frac{n-1}{n}\bar k_P,   
    \end{align}
    where $(a)$ is because the $X_i$-s are i.i.d. $\sim P$ and $(b)$ is because $\kernel(x, x) \equiv \ksup$.
    Similarly, we obtain
    \begin{align}
        & \E_{P^n} \lrang{\mu_{\hat P_{\b X'}}, \mu_{\hat P_{\b X'}}}_{\m H} = \frac{1}{n} \ksup +  \frac{n-1}{n}\bar k_P
        & \E_{P^n  \otimes P^n} \lrang{\mu_{\hat P_{\b X}}, \mu_{\hat P_{\b X'}}}_{\m H} = \bar k_P 
    \end{align}
    By substituting into \eqref{eq:app:EMMDn_expanded}, we get
    \begin{align}
        \E_{P^n \otimes P^n}\MMD_n^2
            &= \frac2n \lrp{\ksup - \bar k_P}.
    \end{align}
\end{proof}

For completeness, we include the following result on the limiting distribution of $n\MMD_n^2$, i.e., of the V-statistic estimator of $\MMD(P, Q)$ under the null hypothesis that $P=Q$---$P$ and $Q$ being probability distributions.
It is mainly based on \citet{gretton2012kernel}, which discusses the U-statistic case.
\begin{theoremA}[Limiting distribution of $\MMD_n$]
\label{thm:mmd_limit}
    Let $\kernel$ be a bounded reproducing kernel for the Hilbert space $\m H$ and define $\MMD_n$ as in \eqref{eq:mmdn_def}.
    Let $\lrc{\lambda_i}_i$ be the eigenvalues of the operator
    \begin{equation}
    \label{app:eq:MMDn_proof_A}
        A: g \mapsto \int_{\m X} \tilde\kernel(x, \cdot)g(x) dP(x),
    \end{equation}
    where $\tilde \kernel(x, x') \coloneqq \lrang{\kernel(\cdot, x) - \mu_P, \kernel(\cdot, x') - \mu_P}_{\m H}$ is the centered kernel.
    Then, 
    \begin{equation}
    \label{app:eq:mmd_limit}
        n \MMD_n^2 \xrightarrow[]{d} 2 \sum\limits_{i=1}^\infty \lambda_i Z_i^2,
    \end{equation}
    where the $Z_i$-s are i.i.d. $\m N(0, 1)$.
\end{theoremA}
\begin{proof}
    Our proof follows \citet[Theorem 12]{gretton2012kernel}  and \citet[Section 6.4.1 Theorem B]{serfling1980approximation}.
    Notice that \citet[Section 6.4.1 Theorem B]{serfling1980approximation} is stated for $\m X = \R$.
    However, it is directly derived from the U-statistic equivalent in Chapter 5, which is stated for more general spaces. 
    
    First, we recall that the $\MMD$ does not change when the kernel is centered~\citep[eq. (9)]{gretton2012kernel}.
    Explicitly,
    \begin{align}
    \label{app:eq:mmd_empirical_centered}
        \MMD_n^2 = \MMD(\hat P_\b X, \hat P_{\b X'})^2 
            = \frac{1}{n^2} \sum\limits_{i, j=1}^n \tilde k(X_i, X_j) + \frac{1}{n^2} \sum\limits_{i, j=1}^n \tilde k(X'_i, X'_j)
            - \frac{2}{n^2} \sum\limits_{i, j=1}^n \tilde k(X_i, X'_j),
    \end{align}
    We will now consider the limiting properties of each addend of \eqref{app:eq:mmd_empirical_centered} separately. 

    \textbf{First Two Addends}\;
    Let $V_n \coloneqq \frac1{n^2}\sum\limits_{i,j=1}^n \tilde\kernel\lrp{X_i, X_j}$. 
    In particular, using Serfling's notation, we have: 
    \begin{enumerate}[label=(\roman*)]
        \item A functional $T(P) = \E_{P \otimes P} \tilde\kernel(X, X') \equiv 0$, with estimator $T(\hat P_n) = V_n$.
        \item A symmetric kernel $h = \tilde\kernel$.
        \item A remainder term $R_n = T(\hat P_n) - T(P) - V_n \equiv 0$
    \end{enumerate}
    Condition $A_2$~\citep[p. 225]{serfling1980approximation} is satisfied:
    \begin{enumerate}[label=(\roman*)]
        \item $\mathrm{var}_P(\tilde\kernel(\cdot, X)) = 0$ by the properties of the centered kernel.
        \item[(i)'] $\mathrm{var}_{P \otimes P}(\tilde\kernel(X, X')) > 0$ if the kernel is not constant. 
        \item[(ii)] $nR_n \equiv 0$. 
    \end{enumerate}
    Then, by boundedness of $\tilde \kernel$ and \citet[Section 6.4.1 Theorem B]{serfling1980approximation},
    \begin{equation}
    \label{app:eq:mmd_1st_limit}
        n V_n \xrightarrow[]{d} \sum\limits_{i=1}^{\infty}\lambda_i a_i^2,
    \end{equation}
    where $a_i\iid\mathcal N(0, 1)$. 
    Similar reasoning shows that the second term converges to $\sum_{i=1}^{\infty}\lambda_i b_i^2$, where $b_i\iid\m N(0, 1)$ and are independent of the $a_i$-s. 

    \textbf{Third Addend}\;
    \citet{gretton2012kernel} eq. (28) shows that
    \begin{equation}
    \label{app:eq:mmd_3rd_limit}
        \frac1n\sum\limits_{i,j=1}^n{\tilde\kernel}\lrp{X_i, X'_j} \xrightarrow[]{d}  \sum\limits_{i=1}^{\infty}\lambda_i a_i b_i.
    \end{equation}

    Recombining \eqref{app:eq:mmd_1st_limit} and \eqref{app:eq:mmd_3rd_limit} and using basic properties of normal and chi-squares, we obtain 
    \begin{equation}
        n \MMD_n^2 \xrightarrow[]{d} \sum\limits_{i=1}^{\infty}\lambda_i (a_i^2 + b_i^2 - 2 a_i b_i) 
        \sim 2\sum\limits_{i=1}^{\infty}\lambda_i Z^2_i,
    \end{equation}
    where $Z_i \iid \mathcal N(0, 1)$, thus concluding the proof. 
\end{proof}

\subsection{Proof of \texorpdfstring{\cref{lemma:KEVmonotonicity}}{Lemma 5}}
\label{app:proofs:KEVmonotonicity}

\begin{proof}
    As $k$ is bounded, $\MMD_m\in[0,M]$ with $M \coloneqq 2\sup_x\sqrt{k(x,x)}$. 
    Therefore, as $\ev_m(P,\cdot) = F_{\MMD_m}$, by the layer cake representation~\citep[Section 1.13]{lieb2001analysis} we have   
    \begin{align}
        \E_{P^{2m}}\MMD_m = \int_0^M\lrp{1-\ev_m(P,y)}\d y.
    \end{align} 
    Hence, we have that
    \begin{align}
      \int_{\R}\lrp{\ev_{n+1}(P,y)-\ev_n(P,y)}\d y = \E_{P^{2n}}\MMD_n-\E_{P^{2n+2}}\MMD_{n+1},
    \end{align}
    and it suffices to show that $\E_{P^{2n+2}}\MMD_{n+1}\leq\E_{P^{2n}}\MMD_n$.
    
    Consider now $(X_i,X'_i)_{i=1}^{n+1} \iid P\otimes P$ and define, for every $i\in[n+1]$, $Y_i \coloneqq k(\cdot,X_i)-k(\cdot,X'_i) \in \m H_k$, so that 
    \begin{align}
    \label{eq:auxiliary_in_proof_KEVmonotonicity}
        \MMD_{n+1}\overset{d}{=}\lrn{\frac{1}{n+1}\sum_{i=1}^{n+1}Y_i}_{\m{H}_k}, 
    \end{align}
    where $\overset{d}{=}$ denotes equality in distribution.
    Now, the leave-one-out means $U_j \coloneqq \frac1n \sum_{i \neq j} Y_i$ satisfy
    \begin{align}
        &\lrn{U_j}_{\m H_k}\overset{d}{=}\MMD_n \quad \forall j \in [n+1]
        &\mathrm{and}&
        &\frac{1}{n+1}\sum_{i=1}^{n+1}Y_i = \frac{1}{n+1}\sum_{j=1}^{n+1}U_j,
    \end{align}
    as each $Y_i$ appears in exactly $n$ of the $U_j$-s.
    Finally, by \eqref{eq:auxiliary_in_proof_KEVmonotonicity}, the triangle inequality, and taking expectations, we have
    \begin{align}
        \E_{P^{2n + 2}} \MMD_{n+1} \leq \frac{1}{n+1}\sum_{j=1}^{n+1}\E_{P^{2n}}\lrn{U_j}_{\m{H}_k} = \E_{P^{2n}}\MMD_n
    \end{align}
\end{proof}

\subsection{Proof of \texorpdfstring{\autoref{lemma:epsfromdelta}}{Lemma 9}}
\label{app:proofs:epsfromdelta}

\begin{proof}
    (First implication).
    \begin{align}
        \E_{P \otimes P} d(X, X') \leq \delta 
            \quad\xRightarrow{(a)}\quad g\lrp{\E_{P \otimes P} d(X, X')} \geq g(\delta)
            \quad\xRightarrow{(b)}\quad \E_{P \otimes P} \kernel(X, X') \geq g(\delta).
    \end{align}
    where $(a)$ is due to monotonicity of $g$ and $(b)$ is due to the convexity of $g$, Jensen's inequality, and \cref{ass:interpretability}.

    (Second implication).
    \begin{align}
        \E_{P^{2n}}\MMD_n 
            \ \overset{(a)}{\leq} \ \sqrt{\E_{P^{2n}}\MMD_n^2} 
            \ \overset{(b)}{=} \ \sqrt{\frac2n \lrp{1 - \bar k_P}}
            \ \overset{(c)}{\leq}\ \sqrt{\frac2n \lrp{1 - g(\delta)}}
            \ \overset{(d)}{\leq}\ \sqrt{2\lrp{1 - g(\delta)}},
    \end{align}
    where $(a)$ is by Jensen's inequality, $(b)$ is from \autoref{lemma:EMMDn} and \autoref{lemma:ksupinterpretable}, $(c)$ is from hypothesis, and $(d)$ is from $n \geq 1$.
\end{proof}
    
\subsection{Proof of \texorpdfstring{\cref{lemma:KEVmarkovlowerbound}}{Lemma 10}}
\label{app:proofs:KEVmarkovlowerbound}

\begin{proof}
    Applying Markov's inequality to $\MMD_n^2$ and then \cref{lemma:EMMDn}, we obtain the following lower bound on the external validity achieved:
    \begin{align}
        \ev_n^{\kernel}(P, \eps) &= P^{2n} \lrp{\MMD_n \leq \eps} 
            \geq 1 - \frac{1}{\eps^2} \E_{P^{2n}} \MMD_n^2
                = 1- \frac{2(1- \bar \kernel_P)}{\eps^2 n} 
    \end{align}
\end{proof}

\subsection{Proof of \texorpdfstring{\cref{thm:mmdapproxpowerlaw}}{Theorem 11}}
\label{app:proofs:mmdapproxpowerlaw}

\begin{proof}
    From \cref{thm:mmd_limit} and the continuous mapping theorem~\citep[Theorem 2.3]{vaart1998asymptotic}, letting $Q\coloneqq2\sum_i\lambda_i Z_i^2$ as in \eqref{app:eq:mmd_limit}, we have that 
    \begin{align}
        \sqrt n \MMD_n \xrightarrow[]{d} \sqrt Q.
    \end{align}
    As the cdf of $Q$ is continuous everywhere, we invoke \citet[Lemma 21.2]{vaart1998asymptotic} to get the convergence of the quantile functions.
    Let $q_\alpha(n) \coloneqq F_{\MMD_n}^{-1}(\alpha)$ and $q_\alpha \coloneqq F_{\sqrt Q}^{-1}(\alpha)$, which are well-defined if there exists $\lambda_i \neq 0$. 
    We then have
    \begin{align}
        &\sqrt n q_\alpha(n) \to q_\alpha &\text{and thus} &&\log q_\alpha(n) + \frac12 \log n \to \log q_\alpha.
    \end{align}
    As $\beta_\alpha \coloneqq \log q_\alpha$ does not depend on $n$, the proof is concluded.
\end{proof}

\subsection{Proof of \texorpdfstring{\cref{lemma:bootstrap_consistency}}{Lemma 13}}
\label{app:proofs:bootstrap_consistency}

\begin{proof}
    First, as customary, we rewrite $\MMD_n^2$ as a one-sample second-order V-statistic.
    From \eqref{eq:mmd_empirical} with $n=m$,
    \begin{align}
        V_n \coloneqq \MMD_n^2 = \frac{1}{n^2} \sum_{i, j=1}^n h\lrp{w_i, w_j},
    \end{align}
    where $w_i = (x_i, x'_i)$ and $h\lrp{(x_i, x'_i), (x_j, x'_j)} \coloneqq k(x_i, x_j) + k(x'_i, x'_j) - k(x_i, x'_j) - k(x'_i, x_j)$.
    From~\cref{thm:mmd_limit}, we know that $nV_n$ converges in distribution to $2 \sum_i \lambda_i Z_i^2$, where the $Z_i$-s are i.i.d. standard normals and the $\lambda_i$-s are the eigenvalues of the operator in \eqref{app:eq:MMDn_proof_A}.

    \emph{$h \in \tilde{L}_2$.} Recall that $\tilde{L}_2(\mathbb{Q}^m)$ is the set of $h \in L_2(\mathbb{Q}^m)$ such that $J_Q h \in L_2(\mathbb{Q}^q)$ for every partition $Q$ of $\{1, \dots, m\}$ into $q$ blocks, where $J_Q h$ identifies the arguments of $h$ within each block \citep[eqs.~(3.3)--(3.5)]{arcones1992bootstrap}. Here $m = 2$ and $\mathbb{Q} = P^2$ is the law of $W = (X, X')$, so that $h \in \tilde{L}_2(\mathbb{Q}^2)$ if and only if $\mathbb{E}\, h^2(W_1, W_2) < \infty$ and $\mathbb{E}\, h^2(W_1, W_1) < \infty$. Both hold, as $|h| \leq 4\bar{k}$ with $\bar{k} := \sup_{x, x'} |k(x, x')| < \infty$; in particular, $\|h\|^2_{\tilde{L}_2} \leq 32\bar{k}^2$.
    
    For a distribution $Q \in \m M_1^+(\m X)$ and $w_1, w_2 \in \m X^2$, define the projection 
    \begin{align}
        (\pi^{Q^2} h)(w_1, w_2) 
            = \lrp{\delta_{w_1} - Q^2} \otimes \lrp{\delta_{w_2} - Q^2} h 
            \coloneqq \int_{\m X^2} h(w, w') \ d\lrp{\delta_{w_1} - Q^2} \d\lrp{\delta_{w_2} - Q^2}.
    \end{align}
    Notice that $\pi^{Q^2}$ is $\pi_{r, m}^{Q^2}$ in \citet[eq. (2.3)]{arcones1992bootstrap}, with $r-1=1$ for the order of degeneracy and $m=2$ for the order of $V_n$.
    Letting $W \sim Q^2$ and $W', W_1, W_2$ be i.i.d. copies of $W$, $\pi^{Q^2} h$ takes the form
    \begin{align}
         (\pi^{Q^2} h)(w_1, w_2) = h(w_1, w_2) - \E_{Q^2} h(w_1, W) - \E_{Q^2} h(W, w_2) + \E_{Q^2 \otimes Q^2} h(W_1, W_2). 
    \end{align}
    Now, letting $\mu_{Q}$ be the mean embedding of $Q$ into $\m H_k$, $X\sim Q$, and $X'$ an i.i.d. copy of $X$, we have
    \begin{align}
        \E_{Q^2} h(w_1, W) 
            &= \E_Q k(x_1, X) + \E_Q k(x_1', X') - \E_Q k(x_1, X') - \E_Q k(x'_1, X) \\
            &= \mu_Q(x_1) + \mu_Q(x_1') - \mu_Q(x_1) - \mu_Q(x_1') = 0
    \end{align}
    Similarly, $\E_{Q^2} h(W, w_2) = 0$.
    For the final term, letting $W_i = (X_i, X_i')$ for $i=1, 2$ we have, since $X_1, X_1', X_2, X_2'$ are i.i.d., 
    \begin{align}
        \E_{Q^2 \otimes Q^2} h(W_1, W_2) 
            &= \E_{Q \otimes Q} k(X_1, X_2) + \E_{Q \otimes Q} k(X_1', X_2') - \E_{Q \otimes Q} k(X_1, X_2') - \E_{Q \otimes Q} k(X_1', X_2)
            = 0.
    \end{align}
    Hence, $\pi^{Q^2} h = h$ for any choice of $Q \in \m M_1^+(\m X)$ and, in particular, for $Q = P_N$.

    To conclude the proof, we invoke \citet[Theorem 3.5, laws (b) and (c)]{arcones1992bootstrap} and obtain that $n\MMD_n^2 = nV_n$ and $\tfrac N2 \widehat \MMD^2_{N/2; N}|\b Y$ converge in distribution to the same random variable for $P^N\aew \ \b Y$. 
\end{proof}

\end{document}

%% file: preamble.tex
\usepackage{amsfonts}
\usepackage{amsmath}
\usepackage{amssymb}
\usepackage{amsthm}

\usepackage{mathtools}
\mathtoolsset{showonlyrefs}

\usepackage{silence}  % ignore annoying warnings
\usepackage{algorithm}  % algorithms
\usepackage{algorithmicx}   % algorithms
\usepackage{algpseudocode}  % algorithms
\usepackage{array}      % ?
\usepackage{bm}         % \bm{1} for indicator function
\usepackage{booktabs}   % tables
\usepackage{cancel}     % ?
\usepackage{caption}    % subfigures
\usepackage{centernot}  % negation of symbols
\usepackage{comment}    % comment env
\usepackage{enumerate}  % roman indices
\usepackage{enumitem}   % itemize options
\usepackage{graphicx}   % graphics options
\usepackage{hyperref}
\usepackage[capitalise,nameinlink]{cleveref}    % \cref, to be loaded after hyperref
\usepackage{lmodern}
\usepackage{mathtools}  % \MoveEqLeft, \coloneqq
\usepackage{multirow}   % tablespzcharm 
\usepackage{natbib}
\usepackage{nicefrac}   % better symbols for 1/2 etc
\usepackage{pifont}     % ?
\usepackage{rotating}   % sidewaysfigure
\usepackage{setspace}   % line spacing macros
\usepackage{stfloats}   %
\usepackage{subcaption} % subfigure env
\usepackage{tabularx}   % ?
\usepackage{textcomp}   %
\usepackage{thmtools}   % continued and restatable theorems
\usepackage{tikz-cd}    % commutative diagrams
\usepackage{url}        % url
\usepackage{verbatim}   % ?
\usepackage[table,xcdraw,dvipsnames]{xcolor}    % colors  

\usepackage{footnote}
\makesavenoteenv{tabular}
\newcommand{\lrp}[1]{\left( #1 \right)}                     % ()
\newcommand{\lrb}[1]{\left[ #1 \right]}                     % []
\newcommand{\lrc}[1]{\left\{ #1 \right\}}                   % {}
\newcommand{\lra}[1]{\left| #1 \right|}                     % ||
\newcommand{\lrn}[1]{\left\lVert #1 \right\rVert}           % || ||
\newcommand{\lrang}[1]{\left\langle #1 \right\rangle}                  % inner product

\newcommand{\R}{\mathbb{R}}                                 % reals
\newcommand{\N}{\mathbb{N}}                               % naturals
\newcommand{\E}{\mathbb{E}}                      % expectation 
\renewcommand{\P}{\mathbb{P}}                               % prob. distr.
\newcommand{\m}[1]{{\mathcal{#1}}}            % mathcal 
\renewcommand{\d}{\mathrm d}

\NewCommandCopy{\textbaraccent}{\b}   % keep LaTeX's original \b (bar-under accent) for text mode
\ExplSyntaxOn
\RenewDocumentCommand{\b}{m}
  {
    \mode_if_math:TF
      {
        \tl_if_head_is_N_type:nTF {#1}
          { \exp_args:Ne \__mybold:Nn { \tl_head:n {#1} } {#1} }
          { \bm{#1} }
      }
      { \textbaraccent{#1} }
  }
\cs_new_protected:Npn \__mybold:Nn #1#2
  { \token_if_letter:NTF #1 { \mathbf{#2} } { \bm{#2} } }
\ExplSyntaxOff

\newcommand{\sign}{\operatorname{sign}}
\newcommand{\iid}{\overset{\mathrm{iid}}{\sim}}     % iid sample
\newcommand{\ind}{\bm{1}}               % indicator function

\newcommand{\aew}{\textnormal{-a.e.}}      % almost everywhere

\newcommand{\alt}{a}                            % generic alternative
\newcommand{\nalt}{{n_{a}}}                     % number of alternatives
\newcommand{\Ra}{\mathcal{R}_\nalt}             % set of all rankings (over \Alt)
\newcommand{\ra}{r}                              % ranking

\NewDocumentCommand{\kernel}{oo}{k\IfValueT{#1}{_{\textnormal{#1}}}\IfValueT{#2}{^{#2}}}

\newcommand{\ksup}{k_{\operatorname{sup}}}
\newcommand{\MMD}{\operatorname{MMD}}                   % MMD
\newcommand{\eps}{\varepsilon}

\newcommand{\transpose}{\top}

\newcommand{\fixed}{\textnormal{f}}
\newcommand{\random}{\textnormal{r}}

\newcommand{\ev}{\operatorname{E}}

\newtheorem{theorem}{Theorem}          % theorem
\newtheorem{lemma}[theorem]{Lemma}          % lemma

\newtheorem{definition}[theorem]{Definition} 
\newtheorem{assumption}[theorem]{Assumption}

\theoremstyle{remark}                           % remark
\newtheorem{remark}[theorem]{Remark}
\newtheorem{example}{Example}

\renewcommand\thmcontinues[1]{Continued}        % continued

\theoremstyle{plain}
\newtheorem{theoremA}{Theorem}[section]
\newtheorem{lemmaA}[theoremA]{Lemma}

\setenumerate{labelindent=0em,leftmargin=2em,topsep=0cm,partopsep=2mm,parsep=0cm,itemsep=1mm}
\setitemize{labelindent=0em,leftmargin=1.3em,topsep=0cm,partopsep=0cm,parsep=0cm,itemsep=1mm}